\documentclass[a4paper,UKenglish]{lipics-v2018}

\usepackage{microtype}
\usepackage{latexsym,epsfig,graphicx,amsmath,algorithm,algorithmicx,algpseudocode,graphicx,epsfig,multirow,calc,colortbl}
\usepackage{amsthm}
\usepackage{mathtools}
\usepackage{color}
\usepackage{cite}

\theoremstyle{plain}
\def\max{{\rm max}}
\def\min{{\rm min}}
\long\def\longdelete#1{}

\DeclareMathOperator{\prrd}{prrd}

\title{\textbf{Tight Competitive Analyses of Online Car-sharing Problems}}

\author{Ya-Chun Liang}{Department of Industrial Engineering and Engineering Management,
National Tsing Hua University, Hsinchu 30013, Taiwan}{ycliang512@gapp.nthu.edu.tw}{}{}

\author{Kuan-Yun Lai}{Department of Industrial Engineering and Engineering Management,
National Tsing Hua University, Hsinchu 30013, Taiwan}{lj841113@gapp.nthu.edu.tw}{}{}

\author{Ho-Lin Chen}{Department of Electrical Engineering,
National Taiwan University, Taipei 106, Taiwan}{holinchen@ntu.edu.tw}{}{}

\author{Kazuo Iwama}{Academic Center for Computing and Media Studies, Kyoto University, Kyoto 606-8502, Japan}{iwama@kuis.kyoto-u.ac.jp}{}{}

\authorrunning{Y.-C. Liang, K.-Y. Lai, H.-L. Chen, K. Iwama}

\Copyright{Ya-Chun Liang, Kuan-Yun Lai, Ho-Lin Chen, Kazuo Iwama}

\subjclass{Algorithms and Combinatorial Optimization}

\keywords{Car-sharing, Competitive analysis, On-line scheduling, Randomized algorithm}

\supplement{}

\funding{}

\acknowledgements{}

\usepackage{here}
\begin{document}
\nolinenumbers
\maketitle

\renewcommand{\algorithmicrequire}{\textbf{Input:}}
\renewcommand{\algorithmicensure}{\textbf{Output:}}

\newcommand{\il}{I\kern-.07em\ell}
\newcommand{\ir}{Ir}
\newcommand{\gf}{G\kern-.1emf}
\newcommand{\gl}{G\ell}
\newcommand{\gr}{Gr}
\newcommand{\ol}{O\ell}
\newcommand{\of}{Of}
\newcommand{\orr}{Or}
\newcommand{\hll}{H\ell}
\newcommand{\hrr}{Hr}

\newcommand{\red}{\color{red}}
\newcommand{\blue}{\color{blue}}
\newcommand{\black}{\color{black}}

\begin{abstract}

The car-sharing problem, proposed by Luo, Erlebach and Xu in 2018,
mainly focuses on an online model in which there are two locations: 0
and 1, and $k$ total cars.  Each request which specifies its pick-up
time and pick-up location (among 0 and 1, and the other is the
drop-off location) is released in each stage a fixed amount of time
before its specified start (i.e. pick-up) time.  The time between the
booking (i.e. released) time and the start time is enough to move
empty cars between 0 and 1 for relocation if they are not used in that
stage.  The model, called $k$S2L-F, assumes that requests in each stage
arrive sequentially regardless of the same booking time and the
decision (accept or reject) must be made immediately.  The goal is to
accept as many requests as possible.  In spite of only two locations,
the analysis does not seem easy and the (tight) competitive ratio (CR)
is only known to be 2.0 for $k=2$ and 1.5 for a restricted value of $k$, i.e.,
a multiple
of three. In this paper, we remove all the holes of unknown CR's;
namely we prove that the CR is $\frac{2k}{k + \lfloor k/3 \rfloor}$ for all
$k\geq 2$. Furthermore, if the algorithm can delay its decision until
all requests have come in each stage, the CR is improved to roughly
4/3. We can take this advantage even further, precisely we can achieve a CR
of $\frac{2+R}{3}$ if the number of requests in each stage is at
most $Rk$, $1 \le R \le 2$, where we do not have to know the value of $R$ in
advance. Finally we demonstrate that randomization also helps to get
(slightly) better CR's.

\end{abstract}

\newpage
\section{Introduction}
Our problem in this paper is the \emph{online car-sharing problem}.
In car-sharing (not only for cars, but also for other resources like bikes and shuttle-buses), there are several service stations in the city, for instance in residential areas and downtown, at popular sightseeing spots, and so on.
Customers can make a request with a pick-up time and place and a drop-off time and place. The decision for accepting or rejecting a request should be made in an online fashion and we want to maximize the profit by accepting as many requests as
possible.
Relocation of 
 (unused) resources 
is usually possible with a much smaller or even negligible costs. (It is seen occasionally that a truck is carrying bikes for this purpose.)
Theoretical studies of this problem have started rather recently and turned out
to be nontrivial even for two locations.

\medskip
\noindent
{\bf Model.}
We basically follow the problem
setting of previous studies by Luo et
al.~\cite{luo2018,luo2018car,luo2019car,LEX_kS2L} with main focus to
that of two locations and $k$ servers (i.e. cars).
The two location are denoted by 0 and 1 
and $k (\geq 2)$ servers are initially located at location 0.
The travel time from 0 to 1 and 1 to 0 is the same, denoted by $t$.
The problem for $k$ servers and two locations is called the $k$S2L problem for short. 

\begin{figure}[h]
\centering
\includegraphics[scale=0.4]{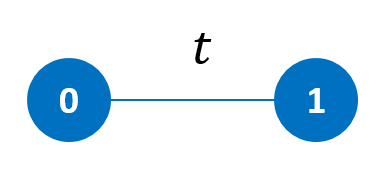}
\caption{The car-sharing problem with two locations}
\label{2location}
\end{figure}

We denote the $i$-th request by $r_i = (\tilde{t_i}, t_i, p_i)$ which
is specified by the \emph{release time} or the \emph{booking time}
$\tilde{t_i}$, the \emph{start time} $t_i$, and the \emph{pick-up
location} $p_i \in \left\{0, 1\right\}$ (the \emph{drop-off location}
is $1-p_i$). If $r_i$ is accepted, the server
must pick up the customer at $p_i$ at time $t_i$ and drop off the
customer at  $1-p_i$  at time $t_i + t$. Suppose for each $r_i$,
$t_i$ is an integer multiple of the travel time between location 0 and
1, i.e., $t_i = vt$ for some $v \in \mathbb{N}$. 
We assume that $t_i
- \tilde{t_i}$ is equal to a fixed value $a$, where $a \ge t$ for all
requests. Without loss of generality, assume $a=t$.
Then we are only interested in
a discrete-time stage,
denoted by 0, 1, 2, $\ldots$.

Each server can only serve one request at a time. Serving a request yields
a fixed positive profit $y$. 
A server used for a request with $p_i=0$ ($p_i=1$, resp.) cannot be used for a
request with $p_i=0$ ($p_i=1$, resp.) in the next stage. 
We allow \emph{empty movement}s, i.e., a server can be moved from one location to
the other
without serving any request. An empty movement spends time $t$,
but takes no cost. The goal of the $k$S2L problem is to maximize the total profit
by serving a set of online requests.
Note that the performance of an online scheduling algorithm is typically evaluated by
\emph{competitive analysis}.
More precisely, the quality of an online algorithm is measured by the worst case ratio, called
\emph{competitive ratio} (CR), which is defined to be the fraction 
 of the profit of the offline optimal algorithm over that of the
online algorithm. The offline algorithm is aware of all requests in
advance. 
If an
online algorithm is randomized, we use expected values for the output
of the online algorithm. 
The online algorithm is called
$1/\delta$-competitive, if for any instances, the profit of the
algorithm is at least $\delta$ times the offline optimal profit. 
So far the current model is exactly the same as $k$S2L-F
in~\cite{luo2018,luo2018car,luo2019car,LEX_kS2L}, although no
randomized cases were discussed in these papers. 

\medskip
\noindent
{\bf Simultaneous Decision Model.}
Recall that in the $k$S2L model, two (or more) inputs with \emph{the
same booking time} still have an order.
Thus we can equivalently think that if $r_1, \ldots, r_d$ are requests 
with booking time $t$, they are coming later than $t-1$ and before or
at $t$, one by one. Each of them should get a
decision (accept or reject) 
immediately before the next request. The adversary can change
$r_i$ after looking at the response of the online algorithm against $r_1,
\ldots, r_{i-1}$. 

This setting sounds reasonable as an online model, but the following
question seems also natural; what if requests with the same booking
time come exactly at the same time, the online player can see
all of them and  can  make decisions all together simultaneously at the booking moment
(equivalently the requests arrive in the same fashion as above but the
player can delay his/her online decisions
until the booking moment).
In this study
we also consider this
new model, denoted by $k$S2L-S. We further extend the
model, assuming that the number of requests with the same booking time
is at most $Rk$ for some constant $R$.
We call the generalized model $Rk$S2L-S.
Notice that having more than $k$ requests at the same location with the same booking time
never helps. Therefore, we only need to study the range $0\leq R\leq 2$
and $k$S2L-S means the special case that $R=2$. Our
algorithm for $Rk$S2L-S is \emph{adaptive} in the sense that it
automatically accommodates the value of $R$, which does not have to be
known in advance.

\medskip
\noindent
{\bf Prior work.}
The car-sharing problem has received a considerable amount of attention in recent years.
Luo et al.~\cite{luo2018} studied the problem with a single server and
two locations with both fixed booking time and variable
booking time. Here ``variable booking time'' essentially means that requests
with start time $t$ may come after requests with start time $t-1$.
 They gave lower bounds on the CR for both fixed and variable
booking time under the positive 
empty movement assumption. 
Later,
Luo et al.~\cite{luo2018car} studied the car-sharing problem with two servers and two locations, i.e. 2S2L.
They considered
only the problem with fixed booking time and proposed an online algorithm which
can achieve a tight bound of two.
Luo et al.~\cite{LEX_kS2L} studied the car-sharing problem with $k$
servers and two locations, for both fixed booking time ($k$S2L-F) and variable
booking time ($k$S2L-V). Namely they showed the CR is
at least 1.5 for all $k$ and at most 1.5 for $k=3i$ for $k$S2L-F and
at least $5/3$ for all $k$ and at most $5/3$ for $k=5i$ for
$k$S2L-V. 
Very recently, Luo et al.~\cite{luo2019car} studied the car-sharing
problem on a star network with $k$ servers as well as
two types of travel time: a unit travel time and an arbitrary travel time.

In comparison with the online setting,
B{\"o}hmov{\'a} et al.~\cite{bohmova2016} considered the offline car-sharing problem
in which all input requests are known in advance.
The objective is to minimize the number of vehicles
while satisfying all the requests.
The offline (i.e. static) problem can be solved in polynomial time.
Another closely related problem is the on-line dial-a-ride problem
(OLDARP),
where objects are transported between given points in a metric space.
The problem has been studied widely.
The goal is to minimize the total makespan~\cite{Ascheuer2000,Antje2017}
or the maximum flow time~\cite{Krumke2005}.
Christman et al.~\cite{christman2018} studied a variation of OLDARP where each request yields a revenue.
Yi et al.~\cite{Yi2005} studied another variation of OLDARP where each
request has a deadline, having a similar flavor as car sharing.

\medskip
\noindent
{\bf Our contribution.} 
Recall that the tight CR of $k$S2L-F is 1.5
for $k=3i$~\cite{LEX_kS2L} and 2 for
$k=2$~\cite{luo2018car}, but open for other $k$'s.
In this paper, we show that it is $\frac{2k}{k+\lfloor k/3\rfloor}$
for all $k\geq 2$ and 1.5 for all $k\geq 2$ if randomization is
allowed.  For $k$S2L-S that allows the online player to delay its
decision, it is shown that we can indeed take this advantage. Namely 
the tight CR for $k$S2L-S is $\frac{2k}{k+\lfloor
k/2\rfloor}$ for all $k\geq 2$ and 4/3 for all $k\geq 2$ if
randomization is allowed. 
For $Rk$S2L-S (we can assume $1\leq R\leq 2$
without loss of generality), it is shown that the CR is
 strictly improved if $R<2$, namely the tight CR
(for randomized algorithms) is improved to $(2+R)/3$. 
Note that if $R= 1.1$ 
(the number of requests at each stage exceeds $k$ by at most $10\%$),
the CR becomes at most $1.034$.

The basic idea of our algorithms is ``greedy'' and ``balanced''. Both notions
have already appeared in~\cite{LEX_kS2L}, but our implementation of
them is significantly different from theirs. More importantly, our
analysis is completely new; namely we use a simple mathematical
induction (augmented by two interesting parameters other than the
profit itself) while a classification of request types
was used in~\cite{LEX_kS2L}.

\begin{table}[h]
\normalsize
\begin{center}
\renewcommand{\arraystretch}{1.15}
\resizebox{\textwidth}{!}{
\begin{tabular}{|c|c|c|c|c|c|c|c|}
\hline
{\bf Problem} & \begin{tabular}[c]{@{}c@{}}{\bf Booking} \\ {\bf Time}\end{tabular} & {\bf Start Time}         & \begin{tabular}[c]{@{}c@{}}  {\bf The Cost Of} \\ {\bf Empty Move}\end{tabular} & \begin{tabular}[c]{@{}c@{}}{\bf Types of} \\ {\bf algorithms}\end{tabular} & \begin{tabular}[c]{@{}c@{}} {\bf Lower} \\ {\bf Bound}\end{tabular} & \begin{tabular}[c]{@{}c@{}} {\bf Upper} \\ {\bf Bound}\end{tabular} & {\bf Reference}        \\ \hline
2S2L    & Fixed                                                   & $t_i$               & $c=y$                                                & Deterministic                 & ---                                                    & 1                                                      & MFCS'18 ~\cite{luo2018car}  \\ \hline
2S2L    & Fixed                                                   & $t_i$               & 0                                                   & Deterministic              & 2                                                      & 2                                                      & MFCS'18 ~\cite{luo2018car}  \\ \hline
2S2L    & Fixed                                                   & $t_i$               & $0<c<y$                            & Deterministic             & 2                                                      & 2                                                      & MFCS'18 ~\cite{luo2018car}  \\ \hline
$k$S2L-F  & Fixed                                                   & $t_i=vt $   for $v \in \mathbb{N}$ & 0                      & Deterministic                                           & 1.5                                                    & $1.5   (k=3i, i\in \mathbb{N})$                                      & ISAAC'18 ~\cite{LEX_kS2L} \\ \hline
$k$S2L-F  & Fixed                                                   &
$t_i=vt $   for $v \in \mathbb{N}$ & 0                      &
Deterministic                                           &  $\frac{2k}{k
  + \lfloor k/3 \rfloor}$                   & $\frac{2k}{k + \lfloor
  k/3 \rfloor}$                                 & this  paper    \\ \hline
$k$S2L-F  & Fixed                                                   &
$t_i=vt $   for $v \in \mathbb{N}$ & 0                      &
Randomized              &  1.5                                               & 1.5                                 & this  paper    \\ \hline
$k$S2L-V  & Variant                                                 & $t_i=vt  $  for $v \in \mathbb{N}$ & 0                       & Deterministic                                          & 1.5                                                    & $1.5   (k=3i, i\in \mathbb{N})$                                      & ISAAC'18 ~\cite{LEX_kS2L} \\ \hline
$k$S2L-V  & Variant                                                 & $t_i=vt   $ for $v \in \mathbb{N}$ & 0                       & Deterministic                                          & 5/3                                                    & $5/3   (k=5i, i\in \mathbb{N})$                                      & ISAAC'18 ~\cite{LEX_kS2L} \\ \hline
\hline
$k$S2L-S     & Fixed
& $t_i=vt   $ for $v \in \mathbb{N}$ & 0                      &
Deterministic                                           &  
$\frac{2k}{k + \lfloor k/2\rfloor}$               &  $\frac{2k}{k + \lfloor k/2\rfloor}$                    & this   paper     \\ \hline
$k$S2L-S     & Fixed                                                   & $t_i=vt   $ for $v \in \mathbb{N}$ & 0                      & Randomized                                           & 4/3                                                     & 4/3                                                    & this   paper
\\ \hline
$Rk$S2L-S ($1\le R \le 2$)     & Fixed                                                   & $t_i=vt$ for $v \in \mathbb{N}$ & 0                      & Randomized                                            & (2+$R$)/3                                                    & (2+$R$)/3                                                    & this   paper     \\ \hline

\end{tabular}}
\end{center}
\caption{Overview of known and new results}
\label{table_summary}
\end{table}

The merit of our new analysis is demonstrated more clearly in $k$S2L-S
than in the original $k$S2L-F.  Therefore we present the results for
$k$S2L-S first and then those for $k$S2L-F; the deterministic case in
Section 2 and the randomized case in Section 3. $Rk$S2L-S is discussed
in Section 4, where we introduce two magic numbers calculated from the
number of requests in each stage.  Finally all matching lower bounds are
given in Section 5.

\section{Deterministic algorithms}

As mentioned in the previous section, we first discuss the basic GBA that
works for the $k$S2L-S model and then its accept/reject version
that works for the original $k$S2L-F model. The analysis for the
former will carry over to that of the latter pretty well. The
following table summarizes our notations which are used in the rest of
the paper.

\begin{table}[H]
\begin{tabular}{ll}
\textbf{Notation} &                                                                         \\
$k$       &The number of total servers \\
$(0, 1)$:           & Requests from location 0 to 1                                         \\
$(1, 0)$:           & Requests from location 1 to 0
\\
$I\kern-.07em\ell_i$:             & The number of (0,1)'s requested in stage $i$ with start time $i$ \\
$Ir_i$:             & The number of (1,0)'s requested in stage $i$ with start time $i$         \\
$G\ell_i$:           & The number of (0,1)'s accepted by the algorithm in stage $i$             \\
$Gr_i$:           & The number of (1,0)'s accepted by the algorithm in stage $i$             \\
$G\kern-.1emf_i$:           & The number of servers not used, i.e., $k-Gr_i-G\ell_i$ \\
$O\ell_i$:         & The number of (0,1)'s accepted by OPT in stage $i$              \\
$Or_i$:         & The number of (1,0)'s accepted by OPT in stage $i$              \\
$O\kern-.1emf_i$:         & The number of servers not used, i.e., $k-Or_i-O\ell_i$ \\
\end{tabular}
\end{table}

\longdelete{
The idea of this greedy balanced algorithm can be illustrated as follows.
All servers are at location 0 in stage 0.
An arbitrary number of servers can be moved to location 1 if needed in stage 1.
We can think of all servers as ``floating'', available at either
location 0 or 1 in stage 1. Let $k=100$ and suppose requests in stage 1
are $(\il_1, Ir_1)=(100,100)$. Thanks to the floating
servers, we can assign 100
servers freely to locations 0 and 1, 
and hence accepted requests, $(\gl_1,\gr_1)$, can be anything like 
(100,0), (75,25) or (0,100). 
However, if (100,0) is selected, then the server allocation in stage 2
is $\bigl[0,0,100\bigr]$.
Here we use notation $\bigl[x,y,z\bigr]$ to indicate $x$ and $z$ servers at locations
0 and 1, respectively and $y$ servers floating. So
the adversary would send
$(\il_2, Ir_2)=(100,0)$ for stage 2, and no servers are
available for the online player. Since the almighty
adversary can select (0,100) in
stage 1, the CR would be 2.
Thus one can easily see that the best thing an algorithm can do is to accept
$(\gl_1, \gr_1)=(50,50)$ in stage 1 to secure a CR of 1.5.
This is the notion of ``\emph{Balanced}''. 
 Note that requests denoted by
$(I\kern-.07em\ell_2, Ir_2)=(100,0)$ have booking time 1, but a
booking time of requests (and when they actually come) 
is no longer important once we know $(\il_i, Ir_i)$ and the server
allocation $\bigl[\gr_{i-1}, \gf_{i-1}=k-\gl_{i-1}-\gr_{i-1}, \gl_{i-1}\bigr]$
in stage $i$. Thus we will simply use ``requests for stage
$i$'' without mentioning their booking time.
}

Suppose there are $x$ and $x'$ servers at location 0 and $y$ and $y'$ servers at location 1 at time $i$, 
where $x$ servers will serve (0,1)'s with start time $i$ but $x'$ servers are not used.
Similarly, $y$ servers will serve (1,0)'s but $y'$ servers not.
Then at time $i+1$ we can use $x$ servers for (1,0)'s and $y$ servers for (0,1)'s. 
Furthermore, $x'+y'$ servers are available for requests of both directions since the requests with start time $i+1$ come at time $i$ and we can move $x'+y'$ servers to whichever locations as we like. 
Now we introduce {\em stages} and say we have $x$~$(=G\ell_i)$ servers at location 1, $y$~$(=Gr_i)$ servers at location 0,
and $x'+y'$ $(=G\kern-.1emf_i)$ ``floating'' servers in stage $i+1$. 
We also denote the {\em server allocation} at time $i+1$ as $\bigl[y,x'+y',x \bigr]$ using the new notation $\bigl[-,-,- \bigr]$. 
Note that we need to have artificial $Gr_0, G\ell_0$ and $G\kern-.1emf_0$ for stage 1 whose values are 0, 0, and $k$, respectively. 

Now the idea of this new greedy balanced algorithm can be illustrated as follows.  
Let $k=100$ and suppose requests in stage 1 are $(\il_1, Ir_1)=(100,100)$. 
As described above, the server allocation in stage 1 is $\bigl[0,100,0\bigr]$, and hence accepted requests, $(\gl_1,\gr_1)$, can be anything like (100,0), (75,25) or (0,100). 
However, if (100,0) is selected, then the server allocation in stage 2 is $\bigl[0,0,100\bigr]$ and the adversary would send $(\il_2, Ir_2)=(100,0)$ for stage 2, 
by which no servers are available for the online player.  Since the almighty adversary can select (0,100) in
stage 1
and he/she can accept all the requests in stage 2, 
the CR would be 2.
Thus one can easily see that the best thing an algorithm can do is to accept $(\gl_1, \gr_1)=(50,50)$ in stage 1 to secure a CR of 1.5.
This is the notion of ``\emph{Balanced}''.
Note that requests denoted by $(I\kern-.07em\ell_2, Ir_2)=(100,0)$ have booking time 1, 
but a booking time of requests (and when they actually come)
is no longer important once we know $(\il_i, Ir_i)$ and the server allocation $\bigl[\gr_{i-1}, \gf_{i-1}=k-\gl_{i-1}-\gr_{i-1}, \gl_{i-1}\bigr]$ in stage $i$. 
Thus we will simply use ``requests for stage $i$'' without mentioning their booking time.

What if $(I\kern-.07em\ell_1, Ir_1)=(60,20)$? In this
case, 
 $(\gl_1, \gr_1)=(60,20)$ 
is the best, i.e., the strategy is
a simple ``\emph{Greedy}'' one. If $(I\kern-.07em\ell_1, Ir_1)=(100,30)$, our
selection is 
 $(\gl_1, \gr_1)=(70,30)$, 
namely ``Greedy'' but as
``Balanced'' as possible. Algorithm~\ref{gba} realizes this idea
almost as it is and it will also be a core of all the subsequent
algorithms in this paper.

\begin{algorithm}
  \caption{GBA($k$): Greedy Balanced Algorithm}
  \label{gba}
  \begin{algorithmic}[1]
  \Require $\il_i$ and $Ir_i$ are the numbers of (0,1)'s and (1,0)'s in stage
  $i$, respectively.
           An integer $k$ is the number of total servers. 
 The server allocation at the beginning of stage $i$ is
$\bigl[\gr_{i-1},\gf_{i-1}, \gl_{i-1}\bigr]$, namely 
$\gr_{i-1}$ and $\gl_{i-1}$ servers at locations 0 and
           1, respectively and $\gf_{i-1}=k-\gl_{i-1}-\gr_{i-1}$ floating servers.
  \Ensure 
 $\gl_i$ and $\gr_i$ are the numbers of accepted (0,1)'s and
          (1,0)'s, respectively.
  \If{$\gr_{i-1}+\gf_{i-1} \leq \lfloor k/2 \rfloor$ or $\il_i \leq \lfloor k/2 \rfloor$}
  \State $\gl_i \leftarrow \min\{\il_i, \gr_{i-1}+\gf_{i-1}\}$; $\gr_i
  \leftarrow\min
  \{\ir_i,\gl_{i-1}+\gf_{i-1}, k-\gl_i \}$;
  \Else
   \If{$\gl_{i-1}+\gf_{i-1} \leq \lfloor k/2 \rfloor$ or $\ir_i \leq \lfloor k/2 \rfloor$ }
   \State $\gr_i\leftarrow\min\{\ir_i, \gl_{i-1}+\gf_{i-1}\}$;\; $\gl_i\leftarrow\min \{\il_i, \gr_{i-1}+\gf_{i-1}, k-\gr_i \}$;
   \Else
   \State  $\gr_i\leftarrow \lfloor k/2 \rfloor$;\; $\gl_i\leftarrow\lceil k/2 \rceil$;
   \EndIf
  \EndIf
  \State \Return $G\ell_i$ and $Gr_i$
\end{algorithmic}
\end{algorithm}

Recall that floating servers in stage $i$ are actually sit at location 0 or 1
at time $i-1$, say 30 ones at location 0 and 10 at location 1 among 40
floating ones. So if we need 20 floating servers in stage $i$ at
location 1, we need to move 10 servers from location 0 to 1 using the
duration from time $i-1$ to $i$. However, we do not describe this
empty movement in our algorithms since it is easily seen and its cost is free in our model.

Observe that the greedy part of Algorithm~\ref{gba} appears in lines 1
and 2 for (0,1)'s and in lines 4 and 5 for (1,0)'s. If
the condition in line 1 is met, then we accept (0,1)'s until we have
exhausted the servers or the (0,1)'s available in this stage. After
that we accept as many (1,0)'s as possible.
It works similarly for lines 4 and
5. If neither the condition in line 1 nor the one in line 4 is met, we
just split the requests almost evenly in line 7. The following theorem shows
that GBA achieves the optimal $\frac{4}{3}$-competitiveness for all even $k$ and approaches this value when $k$ is a large odd number.
In the rest of this paper, we use ALG to denote an online algorithm
and OPT an offline optimal scheduler in general.

\begin{figure}[t]
\centering
\includegraphics[scale=0.6]{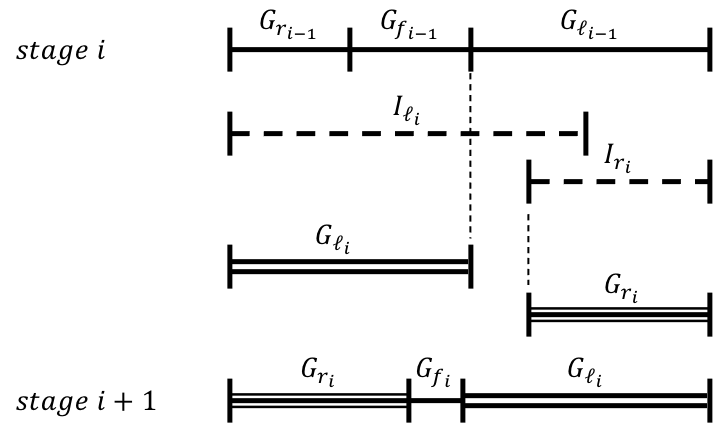}
\caption{Server allocation in GBA}
\label{fig:gba}
\end{figure}

\begin{theorem}
\label{thm:gba}
 GBA is a $1/\delta$-competitive algorithm for $k$S2L-S for any $k\geq 2$, where $\delta = \frac{k+ \lfloor k/2 \rfloor}{2k}$.
\end{theorem}

\begin{proof}

In order to prove the theorem, we consider the following six key values:
$$A_i=\sum_{j=1}^{i} (\gr_j+\gl_j),
\quad  B_i=\sum_{j=1}^{i} (Or_j+\ol_j),$$
 $$X_i=A_i + Gr_i + Gf_i, \quad Y_i=B_i+Or_i + Of_i,$$
$$U_i=A_i + G\ell_i + Gf_i, \quad V_i=B_i+\ol_i + Of_i.$$  
Our goal is to bound $A_i$ by $B_i$. To do so, it
is popular to use a potential function for competitive analysis, which is typically the difference between configurations of ALG and OPT.
In our present case, it may be the difference between server allocations of GBA and OPT.
It turns out, however, that this configuration difference or a similar one is unlikely to work since we still have a freedom for server selection which is not controlled by this difference strongly.
Instead we introduce four parameters, $X_i$, $Y_i$, $U_i$ and $V_i$, which play a key role in our proof.
Note that $X_i$ and $Y_i$ denote the total revenue of GBA and OPT respectively
for the first $i+1$ stages
assuming that the adversary tries to penalize the algorithm choice by
introducing 
 $k$ (0,1)'s in stage $i+1$;  
the last two values, $U_i$ and $V_i$, denote the total revenue of GBA and OPT respectively
for the first $i+1$ stages
assuming that the adversary tries to penalize the algorithm choice by
introducing 
 $k$ (1,0)'s in stage $i+1$. 
Intuitively, GBA
balances the accepted requests in both directions and guarantees that
the CR's in these two instances ($\frac{Y_i}{X_i}$ and
$\frac{V_i}{U_i}$, respectively) are not too large. It turns out that
taking care of these two extreme instances is sufficient to keep the
CR low for all instances.

In order to prove that the algorithm is
$1/\delta$-competitive,
we show that the set of the following inequalities (i) to (iii), denoted by $S(n)$,
$$\text{ (i) } A_n \geq \delta B_n, \quad \text{ (ii) }  X_n \geq \delta Y_n, \quad \text{ (iii) } U_n \geq \delta V_n$$
hold for every $n$ by induction.

For the base case, $n=0$, we have $A_0=B_0=Gr_0=G\ell_0=Or_0=O\ell_0=0$ and
$Gf_0=Of_0=k$. Thus the three inequalities hold since $\delta \leq 1$.

Now the main part of the proof is proving $S(n)$ assuming, as stated in the
induction hypothesis, that $S(j)$
holds for all $0\leq j\leq n-1$. Note that we can rewrite $A_i$,
$B_i$, $X_i$ and so on as follows:
$$A_i=A_{i-1}+G\ell_i+Gr_i, \quad  B_i=B_{i-1}+O\ell_i+Or_i,$$
$$X_i=A_{i-1}+ k + Gr_i, \quad Y_i=B_{i-1}+ k +Or_i,$$
$$U_i=A_{i-1}+k + G\ell_i, \quad V_i=B_{i-1}+k+\ol_i.$$
Since $O\ell_i\leq \min \{ \il_i, Or_{i-1}+Of_{i-1} \}$ and
$Or_i\leq \min \{ \ir_i, O\ell_{i-1}+Of_{i-1} \}$, the
following lemma is obvious, but will be used frequently.

\begin{lemma}
\label{optvalue}
$O\ell_i\leq \il_i$, $\ol_i \leq Or_{i-1}+Of_{i-1}$, $\orr_i \leq \ir_i$, and $\orr_i \leq O\ell_{i-1}+Of_{i-1}$.
\end{lemma}

Now we are ready to prove the theorem.
Suppose line 2 is executed. Then the following (L1) or (L2) holds for
the value of $\gl_i$ and (R1), (R2) or (R3) for the value of
$\gr_i$. Similarly if line 5 is executed, (R1) or (R2) holds for
$\gr_i$ and (L1), (L2) or (L3) for $\gl_i$.
\begin{eqnarray*}
&&\text{(L1)}\: \gl_i=\il_i,\:\: \text{(L2)}\: \gl_i=\gr_{i-1}+\gf_{i-1},\: \text{(L3)}\:
\gl_i =k-\gr_i (\geq \lfloor k/2 \rfloor),\\
&&\text{(R1)}\: \gr_i=\ir_i,\: \text{(R2)}\: \gr_i=\gl_{i-1}+\gf_{i-1},\: \text{(R3)}\:
\gr_i =k-\gl_i (\geq \lfloor k/2 \rfloor).
\end{eqnarray*}
Note that the condition ``$\geq \lfloor k/2 \rfloor$'' in (L3) and (R3)
comes from the conditions in lines 1 and 4, respectively. Now we consider stage $n$ and show that if
(L1), (L2) or (L3) holds, the induction (iii) holds, if (R1), (R2)
or (R3) hold, the induction (ii) holds and if any one of the nine
combinations \{(L1), (L2), (L3)\}$\times$\{(R1), (R2), (R3)\} holds,
(i) holds.  First suppose (L1) holds.  Then since $\ol_n\leq
\il_n$ by Lemma~\ref{optvalue}
\begin{eqnarray*}
&&U_n = A_{n-1}+ k + \gl_n=  A_{n-1}+ k + \il_n,\\
&&V_n = B_{n-1}+k+\ol_n\leq B_{n-1}+k+\il_n.
\end{eqnarray*}
Thus (iii) is true by the induction hypothesis on (i). Similarly for
(L2), i.e.,  by Lemma~\ref{optvalue}
\begin{eqnarray*}
&&U_n = A_{n-1}+ k + \gl_n=  A_{n-1}+ k + \gr_{n-1}+\gf_{n-1}, \\
&&V_n = B_{n-1}+k+\ol_n\leq B_{n-1}+k+ \orr_{n-1}+\of_{n-1}.
\end{eqnarray*}
Thus (iii) is proved by the hypothesis on (ii). Finally for (L3),
\begin{eqnarray*}
&&U_n = A_{n-1}+ k + \gl_n\geq  A_{n-1}+ k + \lfloor k/2 \rfloor, \\
&&V_n = B_{n-1}+k+\ol_n\leq B_{n-1}+k+ k,
\end{eqnarray*}
then use the hypothesis on (i) to claim (iii).

The proof that (R1) or (R2) or (R3) implies (ii) is similar and
omitted.

Finally we show that (i) follows from any combination. Observe
that (L1) and (R1) obviously implies (i)
since no algorithms accept more requests than requested.
All combinations including (L3)
or (R3) are also obvious since GBA accepts $k$
requests. Similarly for (L2) and (R2) when
$\gf_{i-1}=0$ (otherwise impossible). The remaining cases are (L1) and
(R2), and (L2) and (R1).  For the former, by Lemma~\ref{optvalue}
\begin{eqnarray*}
&&A_n = A_{n-1}+ \gl_n+\gr_n = A_{n-1}+\il_n+\gl_{n-1}+\gf_{n-1},\\
&&B_n = B_{n-1}+ \ol_n+\orr_n \leq B_{n-1}+\il_n+\ol_{n-1}+\of_{n-1},
\end{eqnarray*}
and we can use the hypothesis on (iii). (L2) and (R1) is similar and
omitted.

What remains is the case that line 7 is executed. Observe that line 7
gives us $\gl_n\geq \lfloor k/2 \rfloor$, $\gr_n\geq \lfloor k/2
\rfloor$ and $\gl_n+\gr_n=k$. We have already shown that the first one
implies (iii), the second one (ii) and 
the third one means all the servers accept requests and is obviously
enough for (i).
Thus the theorem is proved.
\end{proof}

As seen in GBA and its analysis, a dangerous situation for the
online player is that ALG accepts too many requests of one direction
when it is possible. If ALG knows the total number of requests in each
direction in advance, we can avoid this situation rather easily. Now we discuss
$k$S2L-F, in which ALG does not know the total number of requests in
advance. A simple and apparent solution is to stop accepting requests
of one direction when its number gets to some value, even if more
requests of that direction are coming and could be accepted. In the next algorithm, ARGBA,
we set this value as $2k/3$. It then turns out, a little surprisingly,
that the analysis for Theorem~\ref{thm:gba} is also available for the
new algorithm almost as it is.
\begin{algorithm}[h]
  \caption{ARGBA($k$): Accept or reject GBA}
  \label{argba}
  \begin{algorithmic}[1]
   \Require  The server location is $\bigl[\gr_{i-1},
    \gf_{i-1}=k-\gl_{i-1}-\gr_{i-1},\gl_{i-1}\bigr]$
   at the
  beginning of this stage $i$.
Requests are coming sequentially, each of which, $r$, is (0,1) or
(1,0). $k$ is the total number of servers.
  \Ensure Immediate accept or reject for $r$. $\gl_i$ and $\gr_i$
  for the next server allocation.
  \State $\gl_i\leftarrow 0$; $\gr_i\leftarrow 0$,
  $A\ell_i\leftarrow 0$; $Ar_i\leftarrow 0$ ($A\ell_i$
  ($Ar_i$, resp.) is the number of
  (0,1)'s ((1,0)'s, resp.) received in this stage so far.)
  \While{a new request $r$ comes}
  \If{$r$ is (0,1)}
  \State $A\ell_i \leftarrow A\ell_i+1$;
    \If{$A\ell_i<\gr_{i-1}+\gf_{i-1}$ and $A\ell_i<2k/3$ and
    $\gl_i+\gr_i<k$}
    \State accept $r$; $\:\gl_i \leftarrow \gl_i+1$;
    \Else
    \State reject $r$;
    \EndIf
  \Else $\:\:$ (namely, $r$ is (1,0))
  \State $Ar_i \leftarrow Ar_i+1$;
    \If{$Ar_i<\gl_{i-1}+\gf_{i-1}$ and $Ar_i<2k/3$ and
   $\gl_i+\gr_i<k$}
    \State  accept $r$; $\:\gr_i \leftarrow \gr_i+1$;
    \Else
    \State reject $r$;
    \EndIf
  \EndIf
  \EndWhile
  \State \Return $G\ell_i$ and $Gr_i$
\end{algorithmic}
\end{algorithm}

\begin{theorem}
\label{thm:argba}
ARGBA is a $1/\delta$-competitive algorithm for $k$S2L-F for any
$k\geq 2$, where $\delta=\frac{k + \lfloor k/3 \rfloor}{2k}$.
\end{theorem}

\begin{proof}
Observe that once a (0,1) (similarly for (1,0)) is rejected,
then subsequent ones are all rejected. Let $X_i$ and $Y_i$ be the last
(0,1) and (1,0), respectively, that are accepted and $\gl_i$
and $\gr_i$ be their numbers after lines 6 and 10,
respectively.  Also, let $\il_i$ and $\ir_i$ be the total numbers of
(0,1)'s and (1,0)'s in stage $i$, respectively (only used for analysis).
Suppose the last accepted (0,1) has gone through the
conditions in line 5.
Then, one can see that the next (0,1), if any, is blocked by one of these
conditions and thus
 one of the following four conditions, (L1) through
(L4), is met.  Similarly for (1,0)'s, one of (R1)
 through (R4) is met. 
\begin{eqnarray*}
&&\text{(L1)}\: \gl_i=\il_i,\:\: \text{(L2)}\: \gl_i=\gr_{i-1}+\gf_{i-1},\: \text{(L3)}\:
\gl_i =\lceil 2k/3 \rceil,\\
&&\qquad\qquad\qquad\qquad \text{(L4)}\: \gl_i+\gr_i=k \text{ and }\gl_i\geq
\lfloor k/3\rfloor \text{ and } \gr_i\geq \lfloor k/3\rfloor, \\
&&\text{(R1)}\: \gr_i=\ir_i,\: \text{(R2)}\: \gr_i=\gl_{i-1}+\gf_{i-1},\: \text{(R3)}\:
\gr_i =\lceil 2k/3 \rceil,\\
&&\qquad\qquad\qquad\qquad \text{(R4)}\: \gl_i+\gr_i=k \text{ and }\gl_i\geq
\lfloor k/3\rfloor \text{ and } \gr_i\geq \lfloor k/3\rfloor.
\end{eqnarray*}
Note that the lower bound condition in (L4) and (R4) is correct since
otherwise the second condition in line 5 or 11 should have been met before.

Now consider stage $n$. In a way similar to the proof of
Theorem~\ref{thm:gba}, we can show that one of (L1) to (L4)
(with subscript $n$ replacing $i$) implies the induction (iii).
In fact, the reason is exactly the same for (L1) and (L2) as before.
Using
$\gl_i\geq\lfloor k/3\rfloor$ in (L4) we have
\begin{eqnarray*}
&&U_n = A_{n-1}+ k + \gl_n\geq  A_{n-1}+ k + \lfloor k/3 \rfloor, \\
&&V_n = B_{n-1}+k+\ol_n\leq B_{n-1}+k+ k,
\end{eqnarray*}
and then use the hypothesis on (i) to claim (iii) (recall that
our target CR is relaxed to $\frac{2k}{k + \lfloor k/3
\rfloor}$). (L3) obviously implies $\gl_i\geq\lfloor k/3\rfloor$.

Also we can show, though omitted, that one of (R1) to (R4)
(with subscript $n$ replacing $i$) implies the induction (ii).

We can furthermore show that any one of the 16 combinations of (L1) to
(L4) and (R1) to (R4) implies (i): If a combination includes one of (L3), (L4), (R3) and (R4), then (i) is obvious since ARGBA accepts at least
$\lceil 2k/3 \rceil$ requests in this stage. So we only have to
consider the four combinations ((L1) or (L2), and (R1) or (R2)), and
these cases already appeared in the proof of Theorem~\ref{thm:gba},
which concludes the proof.
\end{proof}


\section{Randomized Algorithms}
Notice that the CR of GBA is 2 when $k=2$. The reason is simple, i.e., the
existence of the ceiling function, 
namely if we can accept a fractional request, our CR would be 4/3.
Of course it is impossible to accept a request by one third, but
it is possible to accept that request with probability $1/3$, which
has the same effect as accepting it by one third in terms of an
expected number.

We define the following function for probabilistic rounding. Let $x$
be a (possibly fractional) non-negative number. Then define
$$\prrd(x)=\begin{cases}
   \lceil x \rceil\quad \text{with probability}\: x-\lfloor x \rfloor\\
   \lfloor x \rfloor\quad  \text{with probability}\: 1-(x-\lfloor x \rfloor)
  \end{cases}
$$
For instance, $\prrd(3.3)$ is 4 with probability 0.3 and 3 with
probability 0.7. $\prrd(3)$ is always 3. Note $E[\prrd(x)] =x$.
Now we are ready to introduce the Probabilistic GBA.

\begin{algorithm}
  \caption{PrGBA($k$): Probabilistic GBA}
  \label{prgba}
  \begin{algorithmic}[1]
  \State The same as Algorithm~\ref{gba} except that line 7 is
  replaced as follows:  $\gr_i \leftarrow \prrd(k/2)$ and $\gl_i
  \leftarrow k-\gr_i$
\end{algorithmic}
\end{algorithm}

\begin{theorem}
\label{thm:prgba}
PrGBA is a $4/3$-competitive algorithm for $k$S2L-S for any $k\geq 2$.
\end{theorem}

\begin{proof}
Observe the induction in the proof of the deterministic case. The base case is fine
with $\delta = 3/4$, and we can keep using this 3/4 unless line 7 is executed. Since the expected
value of $\gr_i$ and $\gl_i$ are both $k/2$ when the modified line 7 is executed, we can remove the ceiling sign
from the description of the algorithm. Thus our new CR is $2k/(k + k/2) = 4/3$.
\end{proof}

\begin{algorithm}
  \caption{PrARGBA($k$): Probabilistic ARGBA}
  \label{prargba}
  \begin{algorithmic}[1]
  \State ARGBA is modified as follows: Suppose the three conditions in
  line 5 are all met and $0<2k/3-A\ell^{j-1}<1$. Then accept
  $r_i^j$ and increase $\gl_i$ if $\prrd(2k/3-A\ell^{j-1})=1$,
  otherwise reject it.
  Similarly for line 11.
\end{algorithmic}
\end{algorithm}

\begin{theorem}
\label{thm:prargba}
PrARGBA is a 1.5-competitive algorithm for $k$S2L-F for any $k\geq 2$.
\end{theorem}
\begin{proof}
The same idea as the proof of Theorem~\ref{thm:prgba}. If this modified
part is executed, the expected value of $\gl_i$ is $2k/3$ and hence the
expected value of $\gr_i$ is $k/3$ if $\gl_i+\gr_i=k$. Thus the worst
case of deterministic ARGBA, $\gl_i=\lceil 2k/3 \rceil$ and $\gr_i=\lfloor k/3
\rfloor$, can be avoided.
\end{proof}


\section{Adaptive GBA}
$k$S2L-S gives the online player the advantage of an advance
knowledge of the number of requests in the current stage. 
GBA does exploit this advantage, but not fully.
Suppose $(I\kern-.07em\ell_1, Ir_1)=(50,100)$. GBA accepts the
same number, 50, of (0,1)'s and (1,0)'s in stage~1. Then the adversary sends
$(I\kern-.07em\ell_2, Ir_2)=(100,0)$, resulting in that only 50 (0,1)'s can
be accepted by GBA in stage 2, but 100 (0,1)'s by OPT which could accept 100 (1,0)'s
in stage~1. Thus the CR in these two steps is 4/3.

Now what about accepting roughly 28.57 (0,1)'s and 71.43
(1,0)'s in stage 1 (recall we can handle fractional numbers
due to randomized rounding)?
Then the best the adversary can do is to provide
$(I\kern-.07em\ell_2, Ir_2)=(100,0)$ or $(0,50)$, in both of which the CR is
$200/171.43\approx 150/128.57 \approx 1.17$, significantly better than 1.5
of GBA. This is the basic idea of our new algorithm, AGBA. These
key values 28.57 and 71.43 are denoted by
$\alpha_i$ and $\beta_i$, respectively.
The ultimate goal of AGBA is
to accept exactly $\alpha_i$ (0,1)'s and $\beta_i$ (1,0)'s while the
ultimate goal of GBA was to accept $k/2$ (0,1)'s and $k/2$ (1,0)'s. If
this goal is unachievable, to be figured out from the
values of $\il_i$, $\ir_i$, $\gl_{i-1}$, $\gr_{i-1}$ and $\gf_{i-1}$,
both algorithms simply turn greedy.

\begin{algorithm}
  \caption{AGBA($k$): Adaptive GBA}
  \label{agba}
  \begin{algorithmic}[1]
  \Require  The same as GBA 
  \Ensure  The same as GBA 
  \If{$R_i=(\il_i+\ir_i)/k \geq 1$}
  \State $\alpha_i=\frac{(1-R_i)k+3\il_i}{2+R_i}$;\:\:
  $\beta_i=\frac{(1-R_i)k+3\ir_i}{2+R_i}$;
  \Else
  \State $\alpha_i=\il_i$; $\beta_i=\ir_i$;
  \EndIf
  \If{$\gr_{i-1}+\gf_{i-1} < \alpha_i$}
  \State $\gl_i\leftarrow\gr_{i-1}+\gf_{i-1}$;\:\:$\gr_i\leftarrow\min \{ \ir_i,\gl_{i-1} \}$;
  \Else
    \If{$\gl_{i-1}+\gf_{i-1} < \beta_i$}
    \State $\gl_i\leftarrow\min \{ \il_i,\gr_{i-1} \}$;\:\:$\gr_i\leftarrow\gl_{i-1}+\gf_{i-1}$;
    \Else
    \State $\gl_i\leftarrow\prrd(\alpha_i)$;
    \State $\gr_i\leftarrow k-\gl_i$;
   \EndIf
  \EndIf
  \State \Return $\gl_i$ and $\gr_i$
\end{algorithmic}
\end{algorithm}

Suppose $\il_i+\ir_i$ is at most $kR_i$ for stage $i$. As
mentioned in a moment, we do not lose generality if $R_i$ is restricted to
$1\leq R_i\leq 2$, under which the CR of AGBA is bounded by $(2+R)/3$, where
$R$ is the maximum value of $R_i$ in the whole stages. We do
not know $R$ in advance and AGBA does not have to, either.

The competitive analysis is given by the following theorem. Note that
if the input in stage $i$ includes more than $k$ (0,1)'s, we can select
an arbitrary subset of size $k$ and similarly for (1,0)'s. This
guarantees that $R\leq 2$ and the case that $\il_i+\ir_i <k$ is
covered by $R=1$. Thus the restriction of $R$, $1\leq R\leq 2$, makes
sense. Also, note that $\alpha_i + \beta_i = k$ whenever $R_i \geq 1$.
AGBA uses $R_i$ but not $R$. $R=\max\{R_i\}$ by definition and so $\delta\leq \delta_i$ (where $\delta_i$ is the local value of $\delta$ in this stage, see its definition at Lemma 8) for all $i$.

\begin{theorem}
\label{thm:agba}
Suppose the number of requests is limited to at most $Rk$ in all
stages for some $R$ such that $1\leq R\leq 2$ and $Rk$ is an
integer. 
 Then AGBA solves $Rk$S2L-S and is 
$1/\delta$-competitive, where $\delta=3/(2+R)$. 
\end{theorem}

\begin{proof}
Because of the probabilistic rounding used in lines 12,
the expected
values of $\gl_i$ and $\gr_i$ are $\alpha_i$ and $\beta_i$,
respectively, if this part is executed. We need two lemmas; the second
one illustrate a tricky nature of $\alpha_i$ and $\beta_i$.

\begin{lemma}
\label{alphabeta}
For any $R_i\geq 0$, $\alpha_i\leq \il_i$ and $\beta_i\leq
\ir_i$. Also, $\alpha_i+\beta_i\leq k$.
\end{lemma}
\begin{proof}
A simple calculation of the formulas in line 2 is enough if the
condition $\il_i+\ir_r\geq k$ is met. Otherwise it is also obvious by
line 4.
\end{proof}

\begin{lemma}
\label{trick}
Let $\delta_i=\min\{1,3/(2+R_i)\}$. Then $k+\beta_i= \delta_i(k+\ir_i)$
and $k+\alpha_i= \delta_i(k+\il_i)$.
\end{lemma}
\begin{proof}
If $R_i<1$ then $\delta_i=1, \alpha_i=\il_i$ and $\beta_i=\ir_i$, so
the lemma is obviously true.  Otherwise, a simple calculation:
$$k+\alpha_i=k+\frac{(1-R_i)k+3\il_i}{2+R_i}=\frac{3k+3\il_i}{2+R_i}=\delta_i(k+\il_i).$$
Similarly for the other.
\end{proof}

The basic strategy of the proof is the same as
Theorem~\ref{thm:gba}. We consider the following four conditions for
the values of $\gl_i$ and $\gr_i$. Suppose line 7 is executed. Then
$\gl_i$ satisfies (L2) and $\gr_i$ satisfies (R1). Similarly if line 10 is executed, (L1) and (R2) hold.
\begin{eqnarray*}
&&\text{(L1)}\: \gl_i=\min\{\il_i,\gr_{i-1}\},\: \text{(L2)}\: \gl_i=\gr_{i-1}+\gf_{i-1}, \\
&&\text{(R1)}\: \gr_i=\min\{\ir_i,\gl_{i-1}\},\: \text{(R2)}\: \gr_i=\gl_{i-1}+\gf_{i-1}.
\end{eqnarray*}
Now consider stage $n$. It is shown that (L1) or (L2) implies (iii) of the induction. For
(L2), the analysis is the same as before and omitted.
For (L1), using Lemma~\ref{optvalue}
for $V_n$, we have
\begin{eqnarray*}
&&U_n = A_{n-1}+k+\gl_n = A_{n-1}+ k+\min\{\il_n,\gr_{n-1}\},\\
&&V_n =B_{n-1}+k+\ol_n\leq B_{n-1}+k+\il_n.
\end{eqnarray*}
If $\min\{\il_n,\gr_{n-1}\}=\il_n$, then we are done using the
hypothesis (i). Otherwise recall that the condition of line 9,
$\gl_{n-1}+\gr_{n-1} < \beta_n$, is met. So we have
$\gr_{n-1}=k-(\gl_{n-1}+\gf_{n-1})\geq k-\beta_n\geq \alpha_n$ (by
Lemma\ref{alphabeta} for the last inequality)
and thus we can use Lemma~\ref{trick} and the hypothesis on (i)
to claim (iii).
The proof that (R1) or (R2) implies
(ii) is very similar and omitted.

Next we prove that each of the four
combinations implies (i). (L1) and (R1), and (L2) and (R2) are obvious
since AGBA is as efficient as OPT or accepts $k$ requests (recall
$\min\{\il_n,\gr_{n-1}\}\geq \alpha_n$ and
$\min\{\ir_n,\gl_{n-1}\}\geq \beta_n$ mentioned above and by Lemma~\ref{alphabeta}). For (L1) and (R2), using
Lemma~\ref{optvalue} and $\alpha_n\leq \il_n$, we have
\begin{eqnarray*}
&&A_n = A_{n-1}+ \gl_n+\gr_n =  A_{n-1}+\min\{\il_n, \gr_{n-1}\}+\gl_{n-1}+\gf_{n-1},\\
&&B_n =B_{n-1}+\ol_n+\orr_n \leq B_{n-1}+\il_n + \ol_{n-1}+\of_{n-1}.
\end{eqnarray*}
Thus the hypothesis on (iii) or $\gr_{n-1}+\gl_{n-1}+\gf_{n-1}=k$
implies that $A_n\geq \delta B_n$. (L2)
and (R1) are similar.

The remaining case is the one that lines 12 and 13 are executed.
If $\gl_n=\alpha_n$, we have $U_n=A_{n-1}+k+\alpha_n$ and thus we can
use Lemma~\ref{trick} to claim (iii) as shown above. Similarly for
$\gr_n$ and (ii). 
(i) is obvious since
$\gl_i+\gr_i=k$, completing the proof.
\end{proof}


\section{CR Lower Bounds}

The CR's given so far are all tight. In this section we prove matching
lower bounds for $k$S2L-S, for $k$S2L-F, for $k$S2L-F with randomization,
and for $Rk$S2L-S with randomization (including for
$k$S2L-S with randomization as a special case). 

\begin{theorem}
\label{lb:gba}
No deterministic online algorithms for the $k$S2L-S problem can achieve a CR of less than $\frac{2k}{k + \lfloor k/2 \rfloor}$.
\end{theorem}
\begin{proof}
Let $\mathcal{A}$ be any deterministic
algorithm. The adversary requests $k$ (0,1)'s and
$k$ (1,0)'s in stage 1.
$\mathcal{A}$ accepts $k_\ell$ (0,1)'s and $k_r$ (1,0)'s.
If $k_\ell \leq \lfloor k/2 \rfloor$, then the adversary requests
$k$ (1,0)'s (and zero (0,1)'s) in stage 2. The profit of $\mathcal{A}$
is $k_\ell + k_r$ in stage 1, and at most 
 $(k-k_\ell-k_r) + k_\ell$  
in stage 2.
Therefore, the total profit of $\mathcal{A}$
is at most
 $k+k_\ell \leq k+\lfloor k/2 \rfloor.$ 
The profit of OPT is $2k$, and the theorem is proved.
If $k_\ell > \lfloor k/2 \rfloor$, then
$k_r \leq \lfloor k/2 \rfloor$.
Now the adversary requests $k$
(0,1)'s in stage 2. The profit of $\mathcal{A}$ and OPT are exactly the same as above and we may omit
 the rest of calculation. Thus the bound is tight.
\end{proof}

\begin{theorem}
\label{lb:argba}
No deterministic online algorithms for the $k$S2L-F problem can achieve a CR of less than $\frac{2k}{k + \lfloor k/3 \rfloor}$.
\end{theorem}
\begin{proof}
Let $\mathcal{A}$ be any deterministic algorithm.  The basic idea is
similar to \cite{LEX_kS2L}. The adversary gives $k$ (0,1)'s (sequentially) for stage 1.  
If $\mathcal{A}$ accepts at most
$\lfloor 2k/3\rfloor$ ones, then the adversary stops his/her requests
and the game ends.  Thus the CR is at least $\frac{k}{\lfloor
2k/3\rfloor}$. Otherwise, if $\mathcal{A}$ accepts $\lceil 2k/3\rceil$
or more, then the adversary gives another $k$ (1,0)'s for stage
1 and $k$ (0,1)'s for stage 2. Since $\mathcal{A}$ has accepted at
least $\lceil 2k/3\rceil$ (0,1)'s in stage 1, $\mathcal{A}$ cannot use
those servers for the (0,1)'s for stage 2. Hence $\mathcal{A}$
can accept at most $k$ requests in stage 1 and at most $k-\lceil 2k/3\rceil$
requests in stage 2, meaning at most $2k-\lceil 2k/3\rceil=k+\lfloor
k/3 \rfloor$ requests in total.  OPT can accept $k$ (1,0)'s in stage 1
and $k$ (0,1)'s in stage 2, i.e., $2k$ in total. Thus the CR is at
least $\frac{2k}{k+\lfloor k/3 \rfloor}$. Since $\frac{2k}{k+\lfloor
k/3 \rfloor} \leq\frac{k}{\lfloor 2k/3\rfloor}$ for all $k$ (this can be
verified by checking for $k=3j$, $k=3j+1$ and $k=3j+2$), the theorem is
proved.
\end{proof}

\begin{theorem}
\label{lb:prargba}
No randomized online algorithms for the $k$S2L-F problem can achieve a
CR of less than 1.5. 
\end{theorem}
\begin{proof} 
The proof is almost the same as that of Theorem~\ref{lb:argba}.
Since the adversary has the full information (other than random values) of ALG,
he/she can compute the expected value of ALG's output. So what we have
to do is just removing the floor and ceiling signs from the previous
proof and considering the resulting numbers as expected values. 
The proof is complete since the previous deterministic OPT has a profit of at least $2k$.
\end{proof}

\begin{theorem}
\label{lb:agba}
 No randomized online algorithms for the $Rk$S2L-S problem can
achieve a CR of less than $\frac{2+R}{3}$. 
\end{theorem}
\begin{proof} 
Let $\mathcal{A}$ be any randomized
algorithm. The adversary requests $\lfloor Rk/2\rfloor$ (0,1)'s and
$\lceil Rk/2\rceil$ (1,0)'s in
stage 1 ($\lfloor Rk/2\rfloor + \lceil Rk/2\rceil=Rk$ by the integrality
condition).
$\mathcal{A}$ accepts $k\ell$ (0,1)'s and $kr$ (1,0)'s.
Let $\alpha= \frac{(1-R)k+3\lfloor
 \frac{Rk}{2}\rfloor}{2+R}$ and if $E[k\ell]\leq \alpha$ (note that the adversary has the full information of
$\mathcal{A}$, so it can compute $E[k\ell]$), then the adversary requests
$k$ (1,0)'s (and zero (0,1)'s) in stage 2. The profit of $\mathcal{A}$
is at most
$$k+E[k\ell]\leq k+\alpha=k+ \frac{(1-R)k+3\lfloor
 \frac{Rk}{2}\rfloor}{2+R}=\frac{3(k+\lfloor
 \frac{Rk}{2}\rfloor)}{2+R}.$$
The profit of OPT is $k+\lfloor Rk/2\rfloor$, and the theorem is proved.
If $E[k\ell]> \alpha$ then let $\beta=\frac{(1-R)k+3\lceil
 \frac{Rk}{2}\rceil}{2+R}$, and it is easy to see that $\alpha+\beta=k$.
Hence we have $E[kr]\leq \beta$ because
$E[k\ell]+E[kr]=E[k\ell+kr]\leq k$. Now the adversary requests $k$
(0,1)'s. The profit of
 $\mathcal{A}$ and OPT are exactly the same as above by replacing $\lfloor
 \frac{Rk}{2}\rfloor$ with $ \lceil\frac{Rk}{2}\rceil$. 
 We may omit
 the rest of calculation. 
 Note that although the two input instances provide a tight lower bound for the competitive ratio, applying Yao's Minimax theorem on any probability distribution over these two input instances does not provide the same tight bound.
\end{proof}

 \section{Concluding Remarks}
\longdelete{
Our analysis fully exploits the notion of ``floating servers,'' and
the arguments of the mathematical induction using supplementary parameters $X_i$
and $U_i$ (and their OPT counterparts), which seems quite powerful and available for future studies of the problem. The current car sharing
model is still elementary, and 
it would be worthwhile to investigate its many extensions.
In particular, relaxation
of the rental period condition should be challenging and important for
practical applications.
}

We have presented a different greedy and balanced algorithm with a new analysis which fully exploits the notion of “floating servers,” and the arguments of the mathematical
induction using supplementary parameters $X_i$ and $U_i$ (and their OPT counterparts).  We believe that the analysis technique is powerful for future studies of many extensions of the car sharing problem, which would be worthwhile to investigate for practical purposes. 
In particular, relaxation of the rental period condition should be challenging and important.

\newpage
\nocite{*}
\bibliography{bibliography}
\bibliographystyle{plain}

\end{document}